\documentclass[preprint]{llncs}

\usepackage{graphicx}        
\usepackage{amssymb,amsmath}
\usepackage{theorem}
\usepackage{multirow}

\newcommand{\REMOVED}[1]{ }

\newcommand{\todo}[1]{
  \textcolor{red}{\footnotesize \textsf{#1}}
}
\renewcommand{\todo}[1]{} 

\newcommand{\topic}[1]{
  \textcolor{Emerald}{\footnotesize \textsf{#1}}
}
\renewcommand{\topic}[1]{} 

\newtheorem{invariant}[theorem]{Invariant}

\title{Lower Bounds for Graph Exploration\\ Using Local Policies}

\author{Aditya Kumar Akash\inst{1} \and S\'andor~P.~Fekete\inst{2} \and   Seoung Kyou Lee\inst{3} \and Alejandro L{\'o}pez-Ortiz\inst{4} \and Daniela Maftuleac\inst{4} \and James McLurkin\inst{3}}
\institute{
IIT Bombay, Mumbai, India.
\email{adityakumarakash@gmail.com}
\and
TU Braunschweig, Braunschweig, Germany.
        \email{s.fekete@tu-bs.de}%
\and
        Rice University, Houston, TX, USA.
        \email{sl28,jmclurkin@rice.edu}%
\and
         Univ. of Waterloo, Waterloo, ON, Canada.
        \email{alopez-o,dmaftule@uwaterloo.ca}%
}
\begin{document}
\maketitle
\begin{abstract}
We give lower bounds for various natural node- and edge-based local strategies for exploring a graph.
We consider this problem both in the setting of an arbitrary graph
as well as the abstraction of a geometric exploration of a space
by a robot, both of which have been extensively studied. We consider
local exploration policies that use time-of-last-visit or alternatively
least-frequently-visited local greedy strategies to select the next step
in the exploration path. Both of these strategies were previously considered by Cooper et al. (2011)
for a scenario in which counters for the last visit or visit frequency are attached to
the edges.  
In this work we consider
the case in which the counters are associated with the nodes, which for the case
of dual graphs of geometric spaces could be argued to be intuitively more natural and likely more efficient.	
Surprisingly, these alternate strategies
give worst-case superpolynomial/exponential time for exploration, whereas the least-frequently-visited
strategy for edges has a polynomially bounded
exploration time, as shown by Cooper et al. (2011).
\end{abstract}



%

\section{Introduction}
\label{sec:Introduction}

\topic{Introduction}
We consider the problem of a mobile agent or robot exploring an arbitrary graph.
This is a well-studied problem in the literature, both in geometric and combinatorial
settings. The robot or agent may wish to explore an arbitrary graph, e.g.
a social network or the graph derived from the exploration of a geometric space.
In the latter case, this is often modeled as an exploration task
in the dual graph, where nodes correspond to rooms or regions, and edges corresponds
to paths from one region to another~\cite{bfk+-tueur-13,flm+-prssr-14,lbf+-estmr-14}.
In either setting, the goal is to explore every node in the graph (i.e., a
corresponding region in space) in the smallest possible worst-case time.
More formally, the question is this: \\
\emph{Given an unknown graph $G$ and a local exploration
policy, what is the time when the last node is visited as a function of
the size of the graph $G$}?



There are several natural local strategy candidates for exploring a graph.
We consider only strategies that use a local policy at each node
for selecting the immediate neighbor that is visited next.
The selection of neighbor can be done using one of the following policies:
(1) {\em Least Recently Visited vertex} (LRV-v), (2) {\em Least Recently Visited edge} (LRV-e),
(3) {\em Least Frequently Visited vertex} (LFV-v), and (4) {\em Least Frequently Visited edge} (LFV-e).

In the strategies above, we assume that each vertex or node holds an associated value,
reflecting the last time it was visited (for the case of least recently visited strategies)
or a counter of the total times it has been explored (for least frequently visited
policies). Then the robot selects the neighboring vertex or adjacent edge
with lowest value, i.e., oldest time stamp or least frequently visited.

Because we are hoping to minimize the time to visit every vertex (the dual of a region
in the geometric space), it would seem
more natural to consider first the LRV-v strategy or failing that, the LFV-v strategy.
However, up until now, the only strategies with known theoretical worst-case
bounds are LRV-e and LFV-e.

However, it has been an open problem whether these natural node-based policies
are efficient. In an experimental study~\cite{mlf+-lpeptr-15}, we consider the task of patrolling
(i.e. repeatedly visiting) a
polygonal space that has been triangulated in a pre-established fashion.
This problem can be modelled as exploration of the dual graph of the triangulation.
This was the original motivation to study the problem of exploring graphs
in general, and the dual graphs of triangulation in particular.
In that paper we sketch an exponential lower
bound for LRV-v.
In this work, we give a superpolynomial lower bound for
LFV-v for exploring a graph. This is in sharp contrast to the edge case,
for which LFV-e has polynomially bounded exploration time. In particular,
we show that there exist a graph on $n$ vertices and $n-1$ edges corresponding to the dual of
a convex
decomposition of a polygon where the convex polygons are fat and of limited area and
such that the exploration time is $\Theta(n^{\sqrt{n}/2} )$ in the worst case.
In the process we show full proofs for
lower bounds for the LRV-v, sketched in the experimental study~\cite{mlf+-lpeptr-15},
and give lower bounds for the LRV-e and worst-case behavior for LFV-e in graphs of degree 3,
thus extending the results by Yanovski et al.~and Cooper et al.~\cite{cik+-drwug-11,yvb+-algo-03} which are shown only
for graphs of higher degree in the so-called ANT model. This model has also been studied by Bonato et al.~\cite{b+-waw-15}
with expected coverage time for random graphs.

\begin{table}[h]
\caption{Summary of results}
\centering
\begin{tabular}{|c|c|c|c|}
  \hline
   Local policy & Graph class & Lower bound & Upper bound \\
   \hline
   \multirow{2}{*}{LRV-v} & general graphs  & $\exp(\Theta(n))$~follows from  Thm~\ref{th:lb.LRV-v} & $\exp(\Theta(n))$ \\
   & duals of triangulations & $\exp(\Theta(n))$~Thm~\ref{th:lb.LRV-v} & $\exp(\Theta(n))$ \\
   \hline
  \multirow{2}{*}{LRV-e} & general graphs  & $\exp(\Theta(n))$~\cite{cik+-drwug-11} & $\exp(\Theta(n))$\\
  & duals of triangulations & $\Omega(n^2)$~Thm~\ref{th:lb.LRV-e} & $\exp(O(n))$ \\
  \hline
  \multirow{2}{*}{LFV-v} & general graphs & $\Omega(n^{\sqrt{n}/2})$~\cite{ksl+-ants-01,mcvw+-distr-05} & $O(n\,\delta^d)$~Thm~\ref{thm:LFVv-delta-d}\\
   & duals of triangulations & $\Omega(n^2)$~Thm~\ref{LFVv-lb} &  $O(n\,\delta^d)$~follows from Thm~\ref{thm:LFVv-delta-d}\\
  \hline
  \multirow{2}{*}{LFV-e} & general graphs & $\Omega(n^2)$ ~follows from  Thm~\ref{LFVe-lb}& $O(m\cdot d)$~\cite{cik+-drwug-11} \\
   & duals of triangulations & $\Omega(n^2)$~Thm~\ref{LFVe-lb} & $O(m\cdot d)$~Cor~\ref{co:ub.LFV-e} \\
   \hline
\end{tabular}
\label{table:results}
\end{table}

\paragraph{Related Work}
\label{sec:RelatedWork} We study policies that require only local information, which can be
maintained by simple devices.  The policy {\em Least Recently Visited}
is known to have worst-case exploration times that are exponential
in the size of the graph, as shown by Cooper et al.~\cite{cik+-drwug-11}.
More recently, the present authors (inspired by empirical considerations) studied LRV-v, LFV-v
and LFV-e in the context of robot swarms and studied the observed average case
using simulations.

The exponential lower bound for LFV-v was shown by Koenig et al.~\cite{ksl+-ants-01}
while Malpani et al. give exponential lower bounds for LRV-e on general graphs \cite{mcvw+-distr-05}.

\paragraph{Summary of Results} In this work we suggest that LFV-e should be the preferred
choice and complement this result by giving (1) an exponential lower bound for the
worst case for LRV-v of triangulations, (2) a quadratic lower bound for
the worst case for LRV-e of triangulations, (3) an exact bound on the maximum frequency
difference of two neighboring nodes in LFV-v, (4) a quadratic lower bound for LFV-v
in graphs of degree 3, (5) a quadratic lower bound for LFV-e
in graphs of degree 3 and, most importantly, (6) a superpolynomial lower bound for the
worst-case of LFV-v when the graph corresponds to a small convex decomposition of a polygon.

\section{Worst-Case Behavior of LRV-e and LRV-v}
\label{sec:LRV}
The worst-case behavior of LRV-e in arbitrary graphs can be
exponential in the number of nodes in the graph, provided we allow a maximum
degree of at least 4. That is, for every $n$, there exists a graph with $n$
vertices in which the largest exploration time for an edge is
$\exp(\Theta(n))$~\cite{cik+-drwug-11}.  Fig.~\ref{fig:reg_graph} illustrates one
such graph (with vertices of degree 4). The starting edge is leftmost in the graph
and the last edge to be visited is the rightmost one. The diamond-like subgraph is
such that when reached by a left-to-right path results in the path
not passing through to the edges on the right on every two-out-of-three
occasions. In this sense the gadget reflects back 2/3rds of all paths.

 If we connect
$\Theta(n)$ such gadgets in series, we will require a total of
$(3/2)^{\Theta(n)}$ paths,
starting from the left for at least one of them to reach the rightmost edge in the series.
\begin{figure}
\centering
\includegraphics[width=0.76\textwidth]{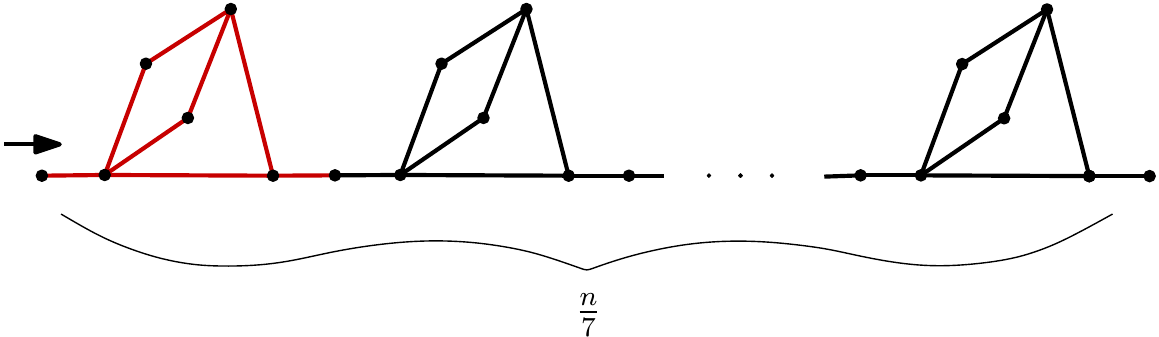}
\caption{Graph with $n$ vertices with a chain of $n/7$
gadgets. A single gadget is colored in red for illustration purposes. Exploring takes exponential time in the worst case~\cite{cik+-drwug-11}.} \label{fig:reg_graph}
\vspace{-3mm}
\end{figure}
Given that our scenario is based on visiting (dual) vertices, it is natural
to consider the worst-case behavior of LRV-v for the special class
of planar graphs of maximum degree 3 that can arise as duals of triangulations.
Until now, this has been an open problem. Moreover, it also makes sense to
consider the worst-case behavior of LRV-e for the same special graph class,
which is not covered by the work of Cooper et al.~\cite{cik+-drwug-11}.

\begin{theorem}
\label{th:lb.LRV-v}
There are dual graphs of triangulations(in particular, planar graphs with $n$ vertices of maximum degree 3), in which LRV-v
leads to a largest exploration time for a node that is exponential in $n$.
\end{theorem}

\begin{proof}
Consider the graph $G_D$ with $n$ vertices in Fig.~\ref{GD-LRV}, which contains
$(n-1)/9$ identical components (each containing 9 vertices) connected in a chain.
This graph is the dual of the triangulation of the polygon in black lines shown in Fig.~\ref{triang-LRV}.
Observe that every vertex has degree at most 3.
We prove the claimed exponential time bound by recursively calculating the time taken to complete one cycle in the transition
diagram shown in Fig.~\ref{LRV}.

\begin{figure}[t]\centering
\includegraphics[width=0.9\linewidth]{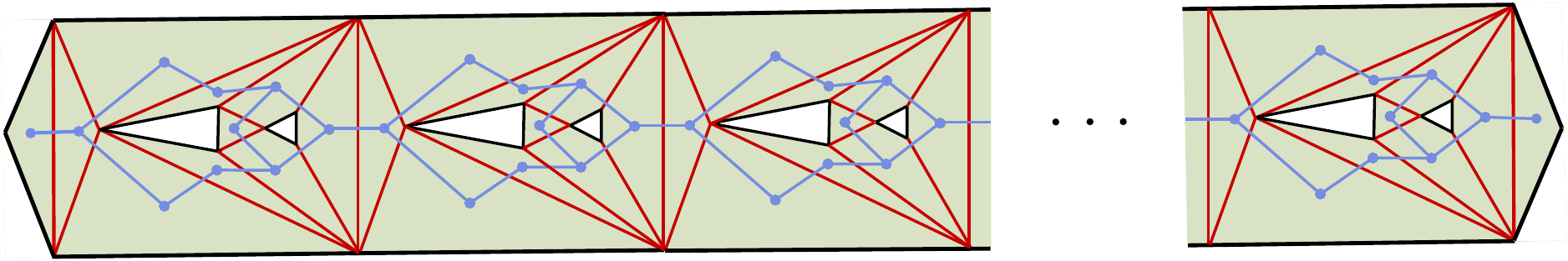}
\caption{This figure depicts (1) a hexagonal polygonal region with holes in black lines (2) its triangulation $G_P$ in red lines
and (3) the dual graph of the triangulation $G_D$ shown in blue lines.} \label{triang-LRV}
 \end{figure}

\begin{figure}[h]\centering
\includegraphics[width=0.7\linewidth]{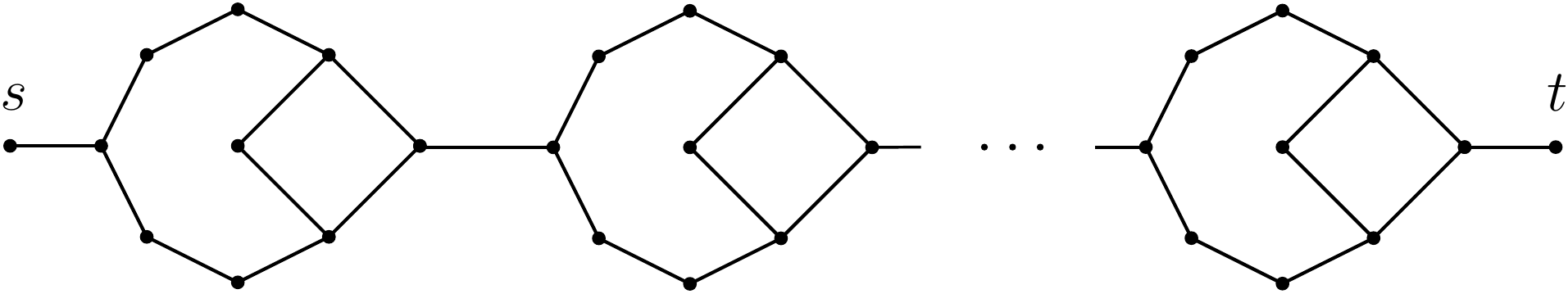}
\caption{The dual graph $G_D$ illustrated in blue in Fig.~\ref{triang-LRV}.} \label{GD-LRV}
 \end{figure}

\begin{figure}[h]\centering
 \includegraphics[width=0.47\textwidth]{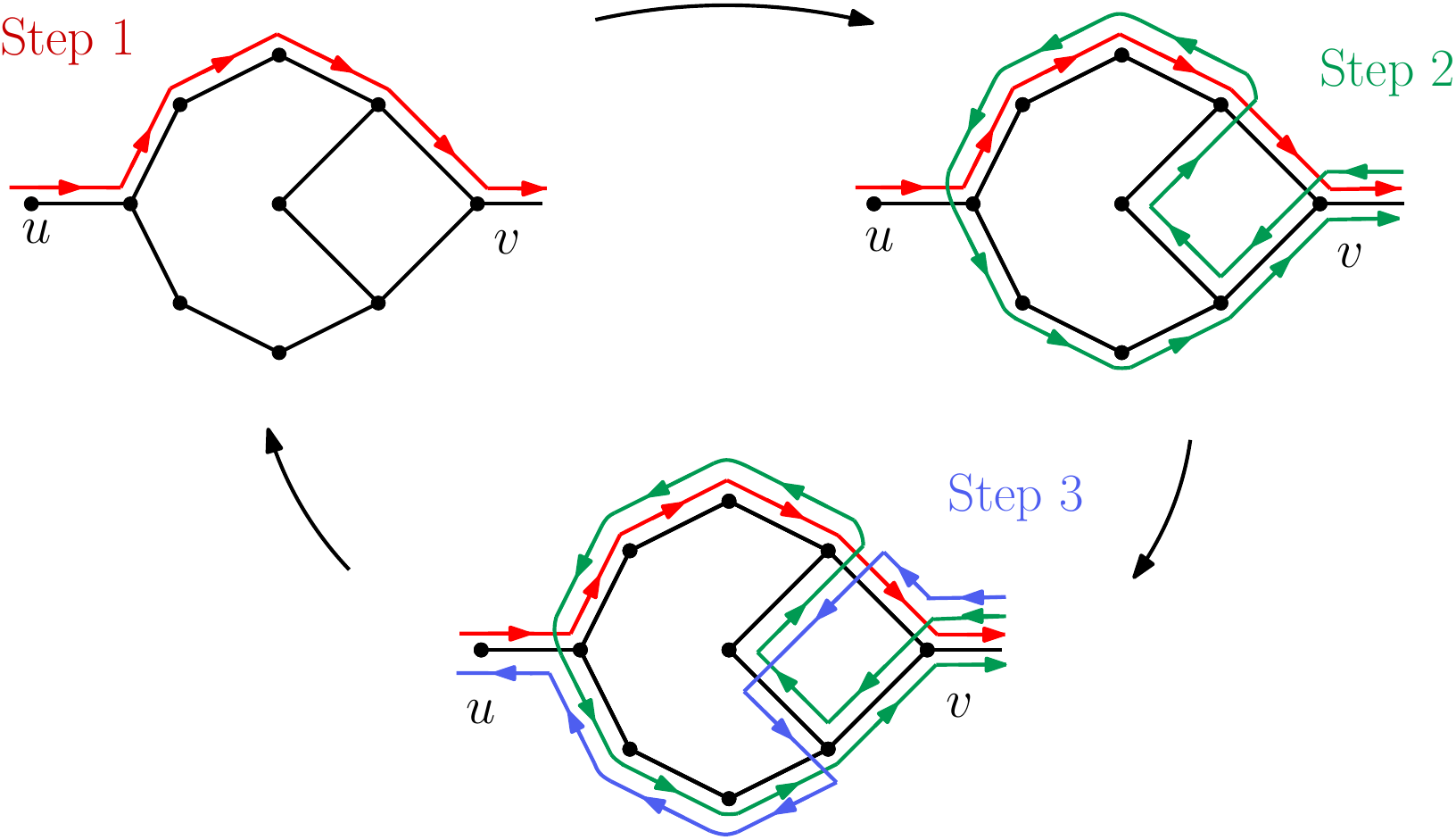}
 \includegraphics[width=0.47\textwidth]{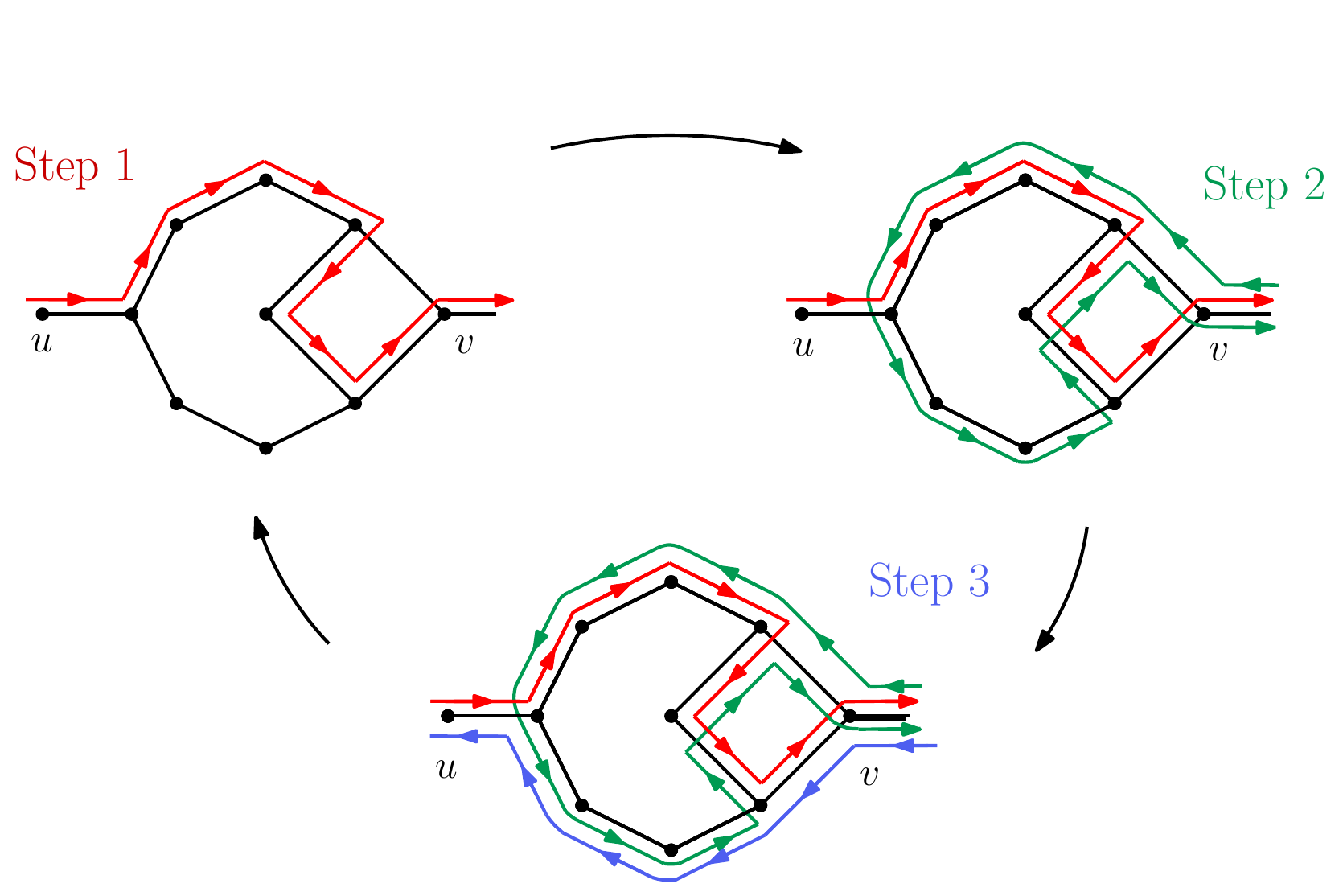}
\caption{Two possible LRV-v alternating paths on each component of the graph $G_D$.} \label{LRV}
 \end{figure}

We monitor the movement of a robot from this situation onwards.
Let $T_n$ denote the time taken to
complete one cycle of $G_D$, i.e., the time taken by a robot to start
from and return to the first vertex of the first component of $G_D$. Similarly,
let $T_i$ denote the cycle which starts on the leftmost vertex $s$ reaching the $i$ gadget on the left-to-right path
and back to $s$.
Hence the graph will be fully explored when we first reach the last component.
This requires three consecutive visits to the next-to-last (penultimate) component in the path, the first
two visits are reflected back to the starting node $s$, and the last goes through.

From the possible paths illustrated in Fig.~\ref{LRV}, we can observe that the vertex $u$ is visited only during the beginning and
end of the cycle, while the vertex $v$ is visited twice in this cycle. It is not hard to check that the summation of visits to all edges
in one component during one cycle is 22. Using this we can see a simple recursion as follows:
$$T_i \geq 22 + 2 \cdot T_{i-1},\quad T_0 = 0$$
Solving this equation, we get $T_{n} \geq 22 \cdot (2^{n}-1)$, and hence the last vertex $t$ in the path is visited
after at least $2\cdot T_{n-1}\geq 22 \cdot (2^{n}-2)$ steps, which is exponential in the number $n$ of nodes of graph $G_D$, as claimed.
\qed
\end{proof}

As it turns out, the lower bound for LRV-e in~\cite{cik+-drwug-11}, i.e., for time stamping vertices instead of edges,
also holds for graphs of max degree 3 as follows.
\begin{figure}[h]\centering
 \includegraphics[width=0.88\textwidth]{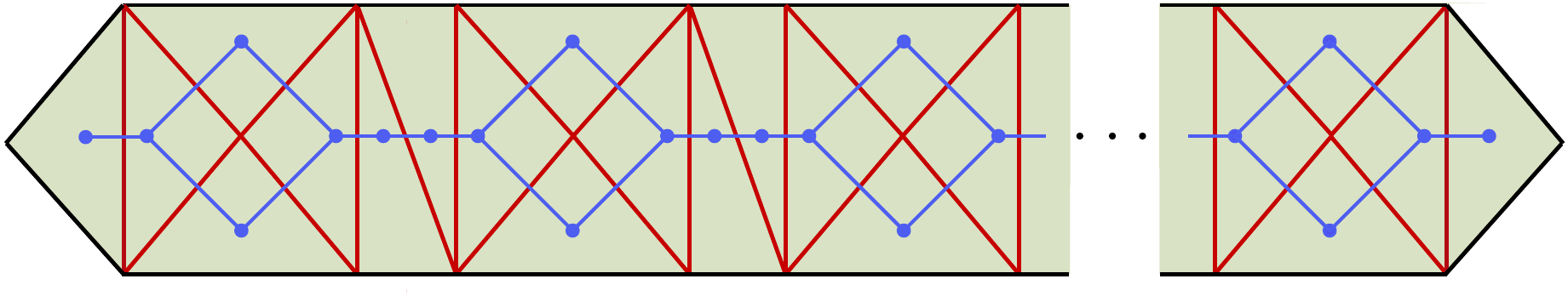}
\caption{This figure depicts (1) a hexagonal polygonal region in black lines (2) its triangulation $G_P$ in red lines
and (3) the dual graph of the triangulation $G_D$ shown in blue lines.} \label{triang}
 \end{figure}

\begin{figure}[h]\centering
\includegraphics[width=0.75\linewidth]{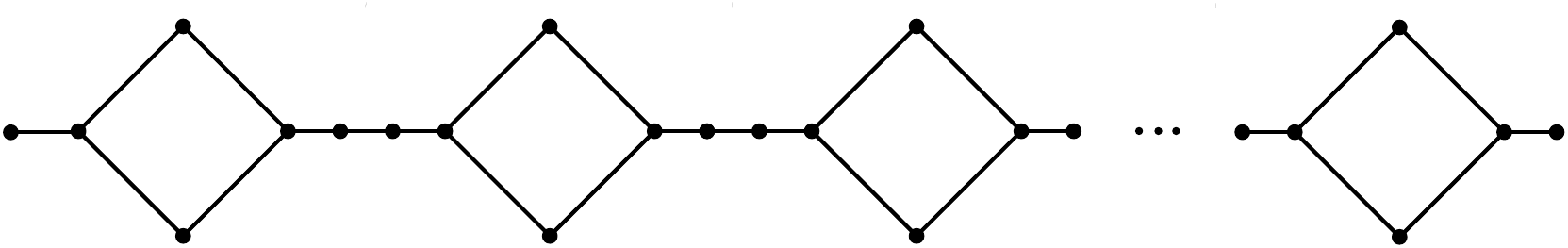}
\caption{The dual graph $G_D$ consisting of a chain of $n/6$ cycles from Fig. \ref{triang}.} \label{GD-LRVe}
 \end{figure}

\begin{theorem}
\label{th:lb.LRV-e}
There are dual graphs of triangulations (in particular, planar graphs with $n$ vertices of maximum degree 3),
in which LRV-e leads to a largest exploration time for a node
that is quadratic in $n$.
\end{theorem}

\begin{proof}
Consider the graph $G_D$ of Fig.~\ref{GD-LRVe} 
which consists of a chain of $n/6$ cycles of length 4 connected in series.
As illustrated in Fig.~\ref{pattern-LRVe}, each component is traversed
initially following the colored oriented paths from step 1 and further
alternating the paths from step $2k$ and $2k+1$, for $k$ positive integer.
When all nodes have time stamp zero we can choose to visit the nodes in any arbitrary
order.\footnote{In general, this property holds whenever there are several neighboring nodes
with the lowest time stamp. For example, in a star starting from the center we visit all neighboring
nodes in arbitrary order until all
of them have time stamp 1. At this point we can once again choose an arbitrary order to visit the
neighbors anew.}

 \begin{figure}\centering
 \includegraphics[width=0.59\textwidth]{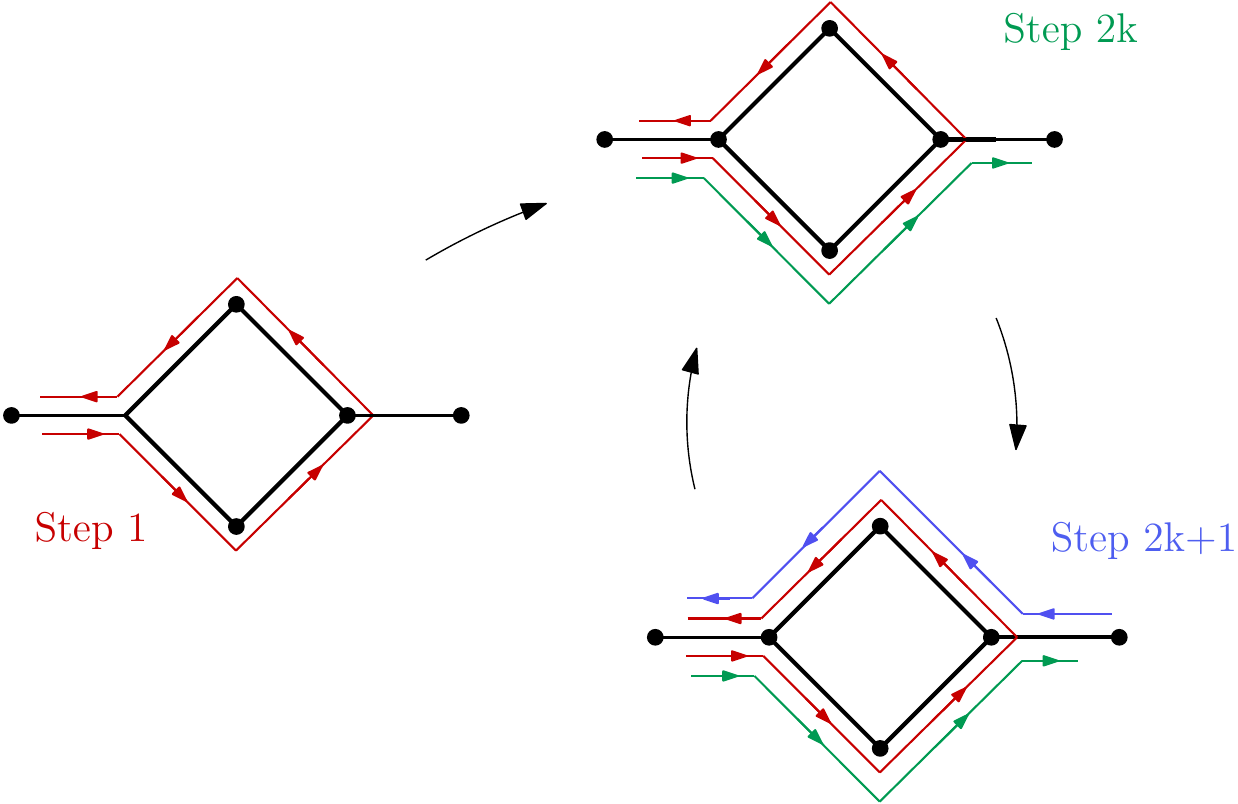}
\caption{LRV-e strategy on each component of the graph $G_D$.} \label{pattern-LRVe}
 \end{figure}
In other words, the first time a component is traversed, the path
changes direction and goes back to the start. The rest of the times
when the component is traversed, the direction does not change.
Thus, in order to traverse the $i$th component in the chain, we need
to traverse the first $i-1$ components in the chain.
The total time, i.e. the number of steps, to reach the rightmost vertex in the chain comes to
$\sum_{i=1}^{n/6}(i-1)=\frac{n(n-6)}{12}=\Theta(n^2)$.
\qed
\end{proof}

\section{Worst-Case Behavior of LFV-v and LFV-e}
\label{sec:LFV}

First, we provide evidence that a polynomial upper bound on the worst-case latency (i.e. time between consecutive
visits)
is unlikely for LFV-v.
We start by showing some interesting properties of graphs explored under LFV-v.
It would seem at first that the nodes in a path followed by the robot form a non-increasing
sequence of frequency values.
This is so as we seemingly always select a node of lowest frequency. However, if
all neighbors of a node have the same or higher frequency, then
the destination node will have strictly larger frequency than the present
node (see Figs.~\ref{path} and \ref{staircase}).

\begin{figure}\centering
 \includegraphics[width=0.47\textwidth]{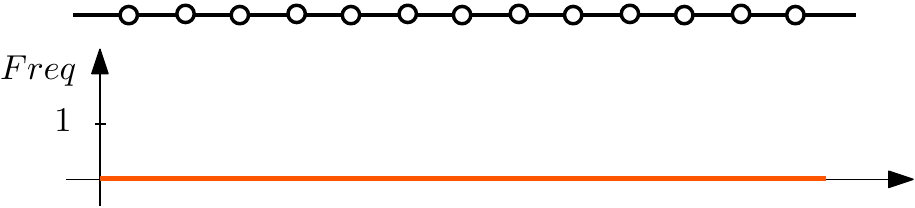}\hspace{0.3cm}
 \includegraphics[width=0.47\textwidth]{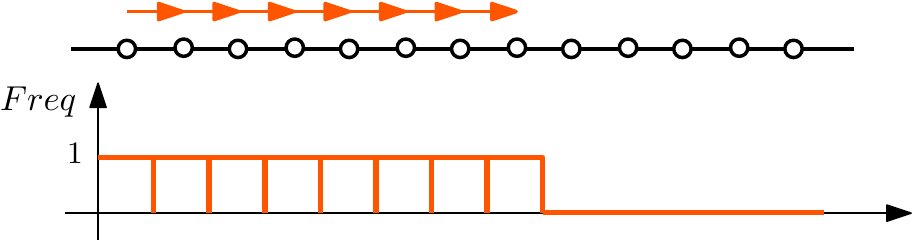}
\caption{A path being traversed from left to right with its frequency histogram below.
Initially all nodes have frequency zero. Then half way through the path traversal
nodes to left have frequency 1 and nodes to the right are still at zero.%
} \label{path}
 \end{figure}

%

\begin{figure}\centering
 \includegraphics[width=0.45\textwidth]{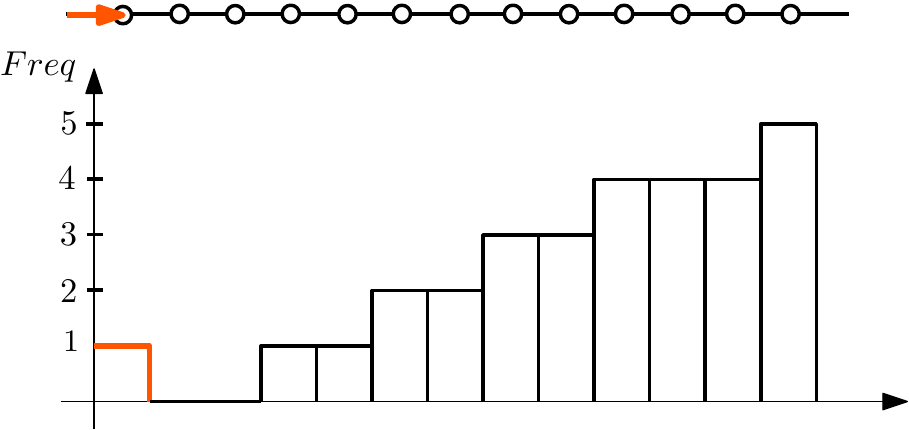} \hspace{0.5cm} \includegraphics[width=0.45\textwidth]{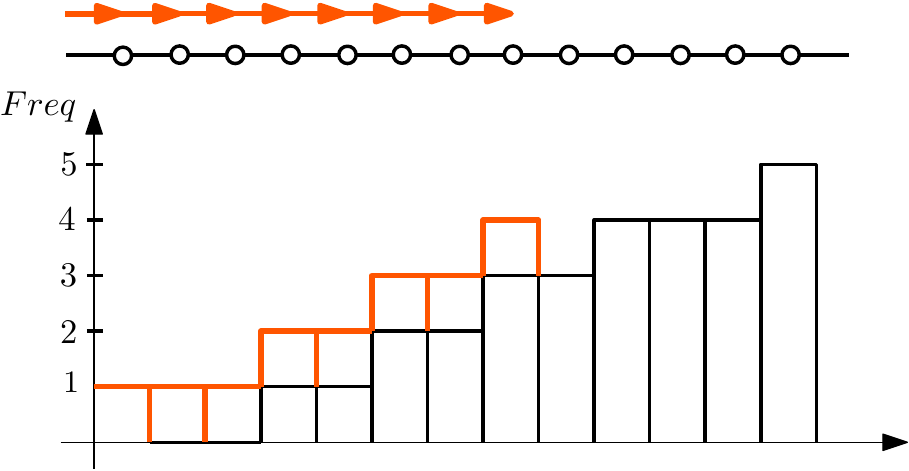}
\caption{A path with a corresponding staircase pattern in the histogram.} \label{staircase}
 \end{figure}

We also observe that it is possible to create dams or barriers by having a
flower configuration  in the path (see Fig.~\ref{petals}). We reach
the center of the flower and then take the loops or petals, thus increasing the
count of the center (see histogram on Fig.~\ref{petals}).
Then the robot moves past the center node of the flower, which forms
a barrier that impedes the robot from traversing from right to left past the
center of the flower, until the count of the nodes to the right of the path has
risen to match that of the barrier.
\begin{figure}\centering
 \includegraphics[width=0.90\textwidth]{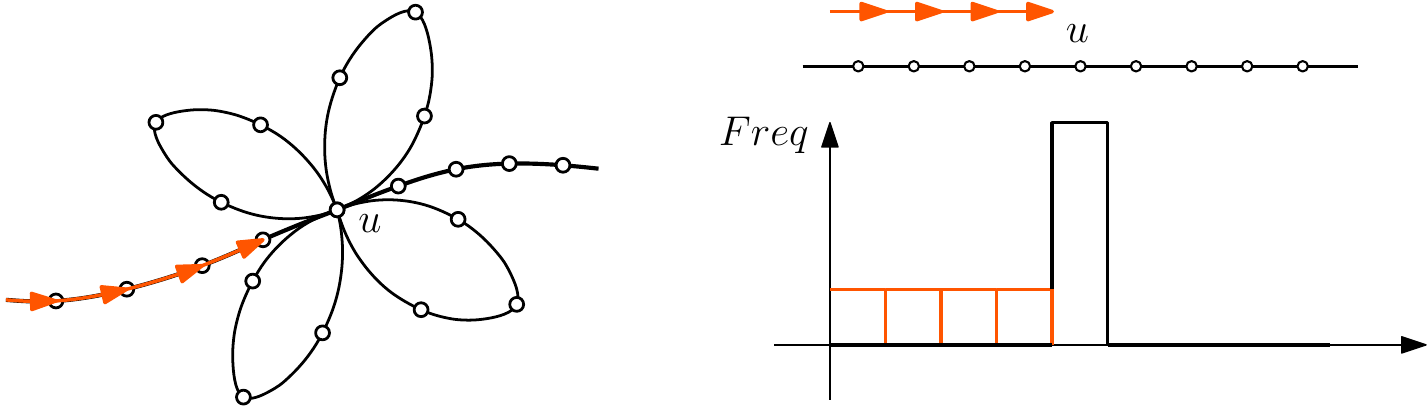}
\caption{A path with a ``flower'' configuration which creates a barrier.%
} \label{petals}
 \end{figure}
With these three basic configurations (path, staircase, flower) in hand, we can combine them to create a
graph in which the starting node $s$ has $\delta(s)$ neighbors as shown in Fig.~\ref{freq}
where we can see that of the $\delta(s)$ neighbors $\delta(s)-1$ have simple paths
leading back to $s$. These paths go via a distinguished neighbor called $u$
which is shared by all the paths from which they connect by a single shared edge to $s$.
Each of the paths is a staircase with barriers (see Figs.~\ref{staircase} and \ref{petals}).
That is for each time we go from $s$ to one of the first $\delta-1$ neighbors
we then climb a staircase up to $u$. Then from $u$ we enter the other
staircases from the ``high'' side until stopped
by a barrier, which makes us return to $u$ and eventually
revisiting $s$ from this last neighbor. This shows the following theorem.

\begin{theorem}
\label{th:ratio}
There exists a configuration for LFV-v in which some neighbors of the starting
vertex have a frequency count of $k$, while the starting point has
a frequency count of $k\,\delta$. Moreover, the value of $k$ can be as high
as $\Theta(n/\delta)$.
\end{theorem}
\begin{figure}\centering
 \includegraphics[width=0.65\textwidth]{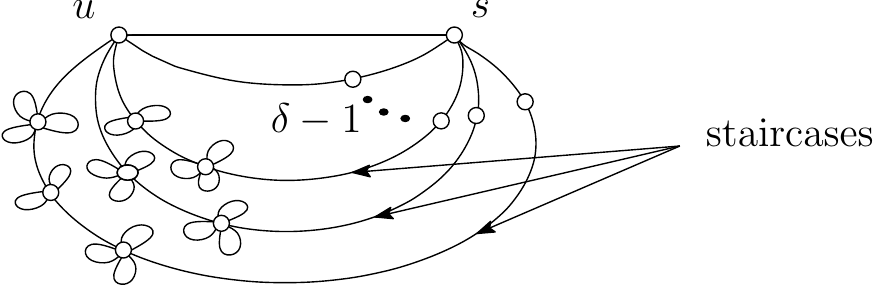}
\caption{A configuration in which the frequency of the starting point $s$ is
much larger than the majority of its neighbors.  } \label{freq}
 \end{figure}
This result provides some indication that the worst-case ratio between
smallest and largest frequency labels of vertices may be exponential,
which would arise if we could construct an example in which the ratio of the
respective frequencies of two neighbors is the degree $\delta$. Observe that $\delta$
can be $\Theta(n)$ in the worst case.
From this it can be shown that at most $\delta^d$ steps are required to explore the graph,
where $d$ and $\delta$ are the diameter and the maximum degree of the graph.

\begin{lemma}
Consider a graph $G_D$ explored using the
strategy LFV-v. Let $g$ denote the frequency of the starting node $s$ at time $t$
and let $\delta(s)$ be the number of neighbors of $s$. Then there are at least $g\bmod \delta$
neighbors with frequency at least $\lfloor g/\delta\rfloor +1$ and the remaining
neighboring nodes have frequency at least $\lfloor g/\delta\rfloor$.
\end{lemma}
\begin{proof}
By induction on $g$. Denote as $g'=g-1$, $f=\lfloor g/\delta\rfloor$,
$f'=\lfloor g'/\delta\rfloor$.

\noindent \emph{Basis of induction.} $g=0.$ In this case $f=\lfloor 0/\delta\rfloor=0$, so we have
trivially at least $(0 \bmod \delta) =0$ nodes with frequency at least 1 and the rest of the nodes have frequency at least 0.
For good measure the reader may wish to prove the case $g=1$.

\noindent \emph{Induction step.} When $g$ increases by one, we have either (1) $f=f'$
or (2) $f=f'+1.$

\noindent In case (1) the robot explores a neighbor with min frequency at
least $f'=f$ whose frequency increases to $f+1$, thus increasing the
number of neighbors with that frequency by 1 (if no such neighbor
exists this means all neighbors already have frequency at least $f+1$
and hence it trivially holds that at least $(g \bmod \delta)$ neighbors have
frequency at least $f+1$).

\noindent In case (2) when $f=f'+1$ we have $\lfloor(g-1)/\delta\rfloor+1=\lfloor g/\delta\rfloor$ which
implies $(g \bmod \delta) = 0$ and $(g-1 \bmod \delta) = \delta-1$. Hence all but one
of the neighbors are guaranteed to have frequency at least $f'+1$
(which is equal to $f$) and there is at most one neighbor with frequency
$f'$ which is the min and gets visited thus increasing its frequency
to $f'+1=f$. This means that now all neighbors have frequency at least $f$
and trivially at least $(\delta(s) \bmod \delta)=0$ neighbors have frequency at
least $f+1$, as claimed.
\qed
\end{proof}

\begin{theorem}\label{thm:LFVv-delta-d}
The highest frequency node in a graph with unvisited nodes, using LFV-v, has frequency
bounded by $\delta^d$, where $\delta$ is the degree of the node and $d$ the
diameter of the graph.
\end{theorem}
\begin{proof}
Consider any shortest path from an unvisited node to the node with highest frequency.
The path is of length at most the diameter $d$ of the graph. In each step
the increase in frequency is at most a factor $\delta$ over the unvisited node
hence the frequency of the most visited node is bounded by $\delta^d$.
\qed
\end{proof}
However, there is no known example of a dual of a triangulation graph
displaying this worst-case behavior.
\begin{theorem}\label{LFVv-lb}
There exist graphs with $n$ vertices of maximum degree 3, in which the largest
exploration frequency for a node, using LFV-v, is $\Theta(n^2).$
\end{theorem}
\begin{proof}
This proof follows the outline of the proof of Theorem \ref{th:lb.LRV-e}
for the same graph $G_D$ represented in Fig.~\ref{triang}.
As illustrated in Fig.~\ref{LFVv}, each of $G_D$'s components is traversed
initially following the colored oriented paths from step 1 and further
alternating the paths from step $2k$ and $2k+1$.

%
 \begin{figure}[h]\centering
 \includegraphics[width=0.6\textwidth]{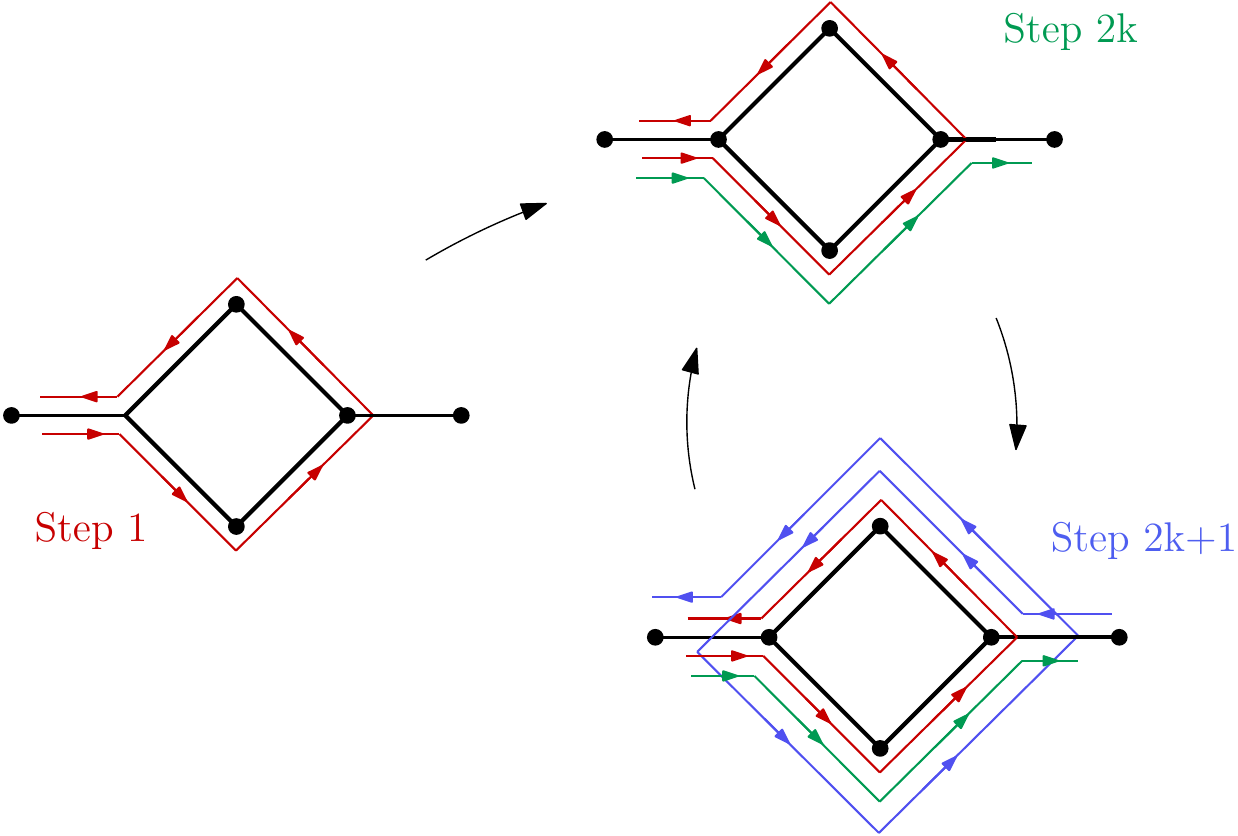}
\caption{LFV-v strategy on each component of the graph $G_D$.} \label{LFVv}
 \end{figure}
In other words, the first time a component is traversed, the path
changes direction and goes back to the start. The rest of the times
when the component is traversed, the direction does not change. Thus, in order to traverse the $i$th component in the chain, we need
to revisit the first $i-1$ components in the chain, which is $\Theta(n^2)$.
\qed
\end{proof}

Note that using LFV-e on the graph shown in Fig.~\ref{GD-LRVe}, each component of the graph is traversed using the exact same
strategy as shown in Fig.~\ref{LFVv} for LFV-v.
\begin{theorem}\label{LFVe-lb}
There exist graphs with $n$ vertices of maximum degree 3, in which
the largest exploration frequency for an edge, using LFV-e, is $\Theta(n^2)$.
\end{theorem}

\begin{theorem}
\label{th:ub.LFV-e}
~\cite{cik+-drwug-11} In a graph $G$ with at most $m$ edges and diameter $d$, the latency of each edge when carrying out
LFV-e is at most $O(m\cdot d)$.
\end{theorem}
This allows us to establish a good upper bound on LFV-e in our setting.
\begin{corollary}
\label{co:ub.LFV-e}
Let $G_D=(V_D,E_D)$ be the dual graph of a triangulation, with $|V_D|=n$ vertices and diameter $d$. Then
the latency of each vertex when carrying out LFV-e is at most $O(n\cdot d)$.
\end{corollary}

\begin{proof}
Since $G_D$ is planar, it follows that $n\in\Theta(m)$, where $m=|E_D|$ is the number of edges
of $G_D$. Because patrolling an edge requires visiting both of its vertices, the claim follows from the
upper bound of Theorem~\ref{th:ub.LFV-e}.
\qed
\end{proof}

We note that this bound can be tightened for regions with small aspect ratio, for which
the diameter is bounded by the square root of the area.

\begin{corollary}
\label{co:aspect}
For regions with diameter $d\in O(\sqrt{n})$, the
latency of each dual vertex when carrying out LFV-e is at most $O(n^{1.5})$.
\end{corollary}

\section{A Graph with Superpolynomial Exploration Time}

Koenig et al. gave a graph requiring superpolynomial exploration time, thus proving the theorem:
\begin{theorem}
\label{co:ub.LFV-v}~\cite{ksl+-ants-01}
LFV-v has worst-case exploration time $\Omega(n^{\sqrt{n}/2})$ on an $n$ vertex graph. This holds even
if the graph is planar and has sublinear maximum degree.
\end{theorem}
We illustrate a similar construction for completeness. The graph is a caterpillar tree where the central path has $\ell+2=\lfloor \sqrt{n} \rfloor$ vertices,
and without loss of generality we assume $\ell$ to be odd (see Fig.~\ref{fig:caterpillar}).
The root which we term node 0, has $b+c+1$ leaves (for some constant $c>10$ and value $b$ to be determined later)
as children plus one edge connecting to the path, for a total degree of $b+c+2$.
The $i$th node in the path has $b-i+1$ leaves as children and is connected
in a path; hence node $i$ has degree $b-i+3$, for $1\leq i \leq \ell$.
The last node in the path has $b+1$ leaves as children and degree $b+2$.

\begin{figure}[h]
\vspace{-0.3cm}
\centering
\hspace{-0mm}\includegraphics[height=.645\linewidth]{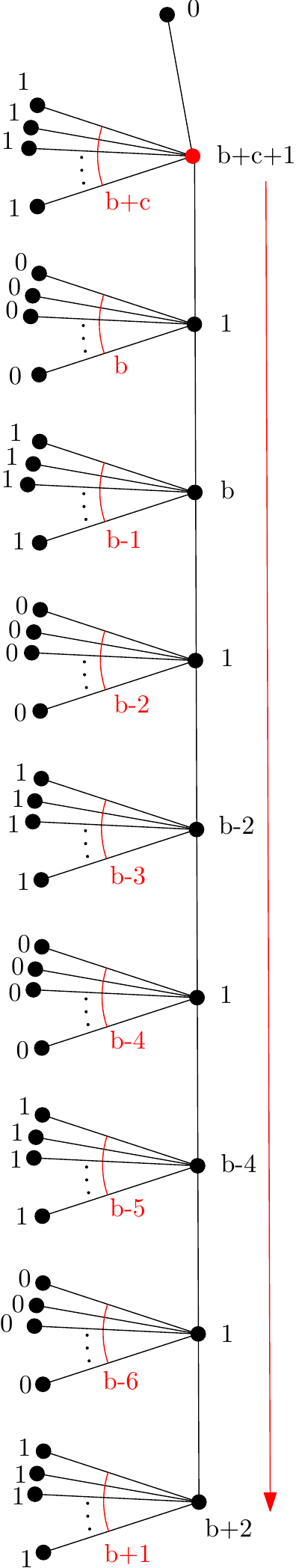}\hspace{-2mm}
\includegraphics[height=.645\linewidth]{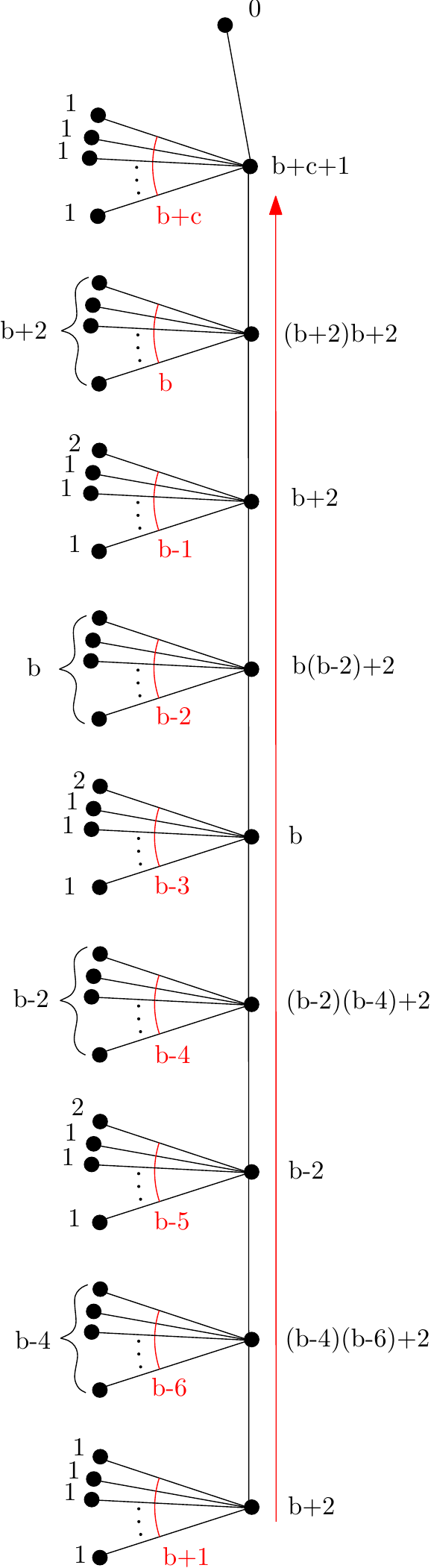}\hspace{-5mm}
\includegraphics[height=.645\linewidth]{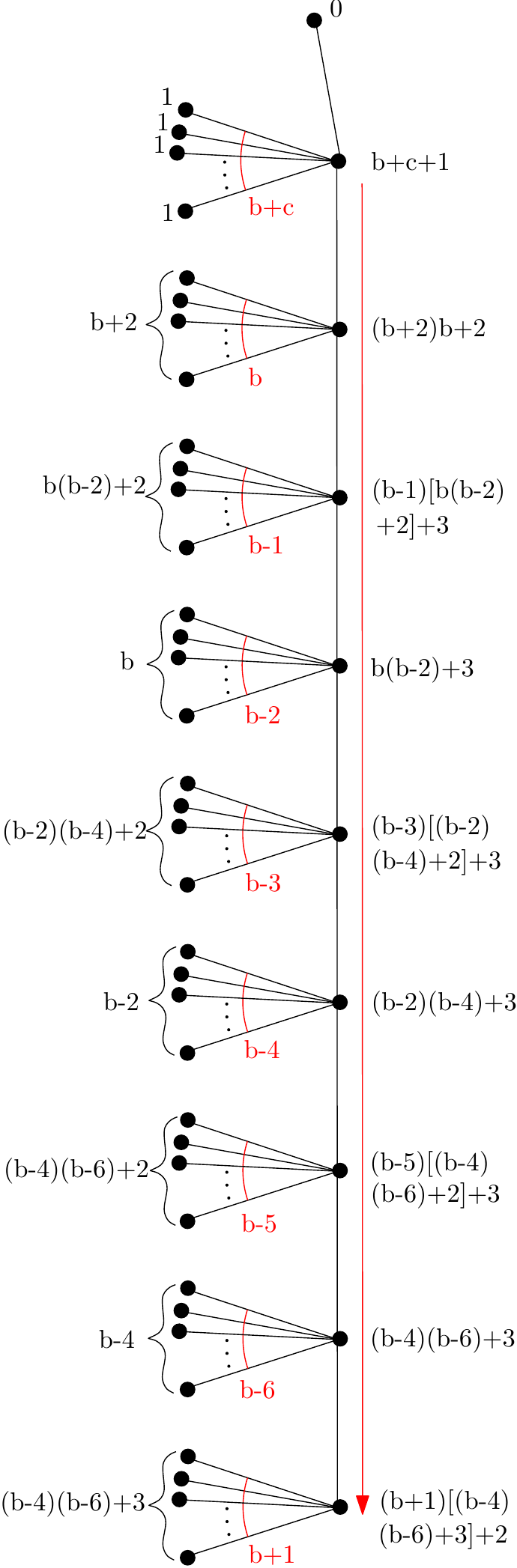}\hspace{-1mm}
\includegraphics[height=.645\linewidth]{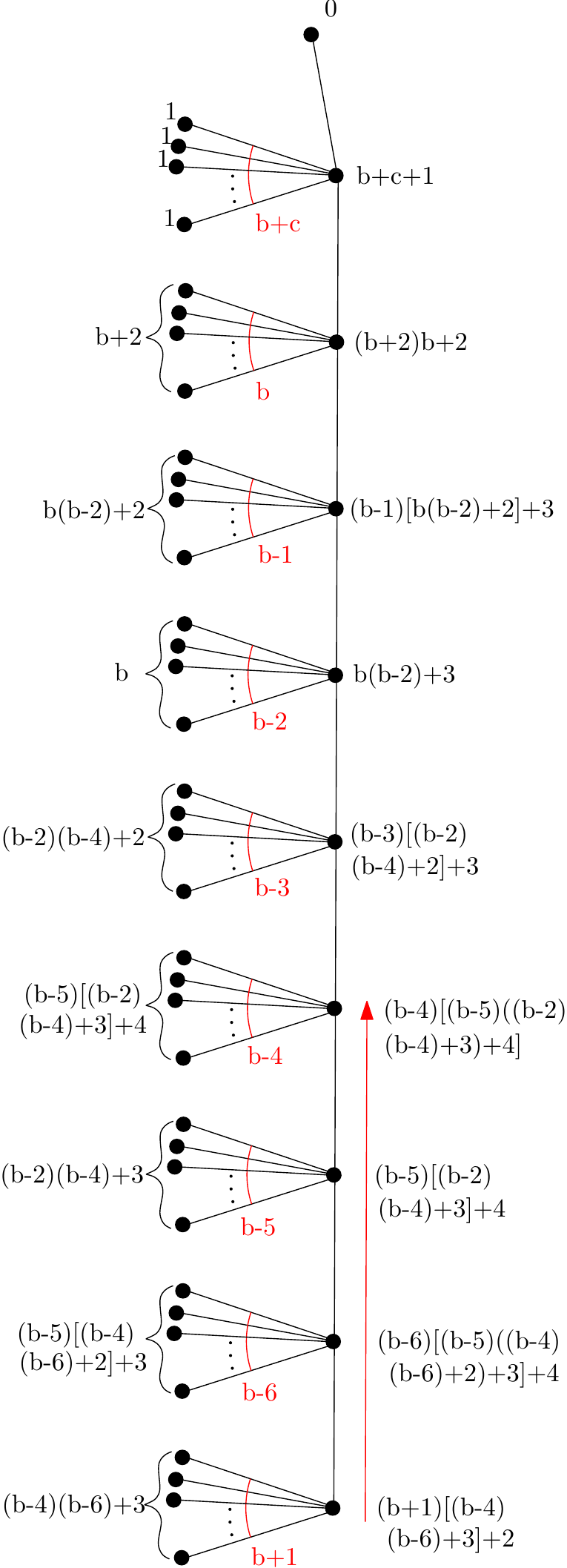}\hspace{-3.5mm}
\includegraphics[height=.645\linewidth]{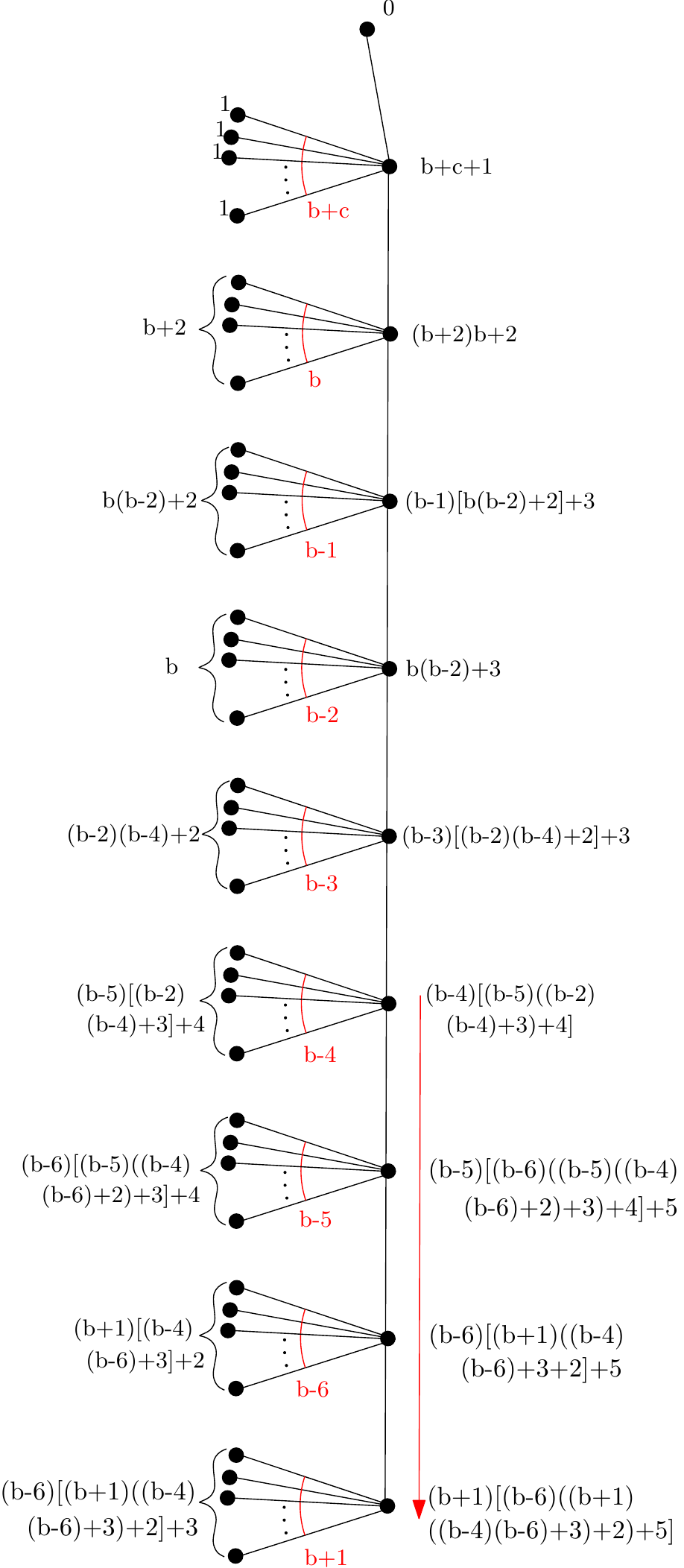}
\caption{
\label{fig:caterpillar}
A caterpillar tree requiring superpolynomial exploration time.}
\end{figure}


We start from the root and as all nodes have visit frequency 0 we can choose to visit the nodes in any arbitrary
order. Recall that this property holds whenever there are several neighboring nodes
with the lowest frequency. From the root we visit all leaves save one, thus increasing the frequency of the root
to $b+c+1$ and then proceed down the path. Then for a node $i$, $1\leq i \leq\ell$ in the path, if $i$ is odd
we arbitrarily follow the path leaving the leafs with frequency 0, while if $i$ is even
we visit all leaves first and then proceed down the path, thus increasing the frequency of node $i$ to $b-i+2$.
In the last node in the path we visit all of its $b+1$ children and then return to it with a final
frequency of $b+2$. More formally
\[
\mbox{freq}(i)=\left\{\begin{tabular}{ll}
  $b+c+1$ &\mbox{ if $i=0$ }\\
  $1$  &\mbox{ if $1\leq i \leq \ell$ is odd }\\
  $b-i+2$ &\mbox{ if $1\leq i \leq \ell$ is even }\\
  $b+2 $ &\mbox{ if $i=\ell+1$ }
\end{tabular}
\right.
\]
At this point the last node in the path has all of its neighbors of degree 1 so we
arbitrarily choose to move up the path as shown in Fig.~\ref{fig:caterpillar} (column 2). Observe that
node $\ell$ in the path has $b-\ell+1$ leaves with frequency 0 and the two neighbors in the
path have frequency $\Theta(b)$, so the robot will visit the leaves $\Theta(b)$ times until
it matches the smallest frequency of the two neighbors in the path, which in this
case is the node $\ell-1$ above with frequency $b-\ell+3$. This gives
a total frequency of
$(b-\ell+1)(b-\ell+3)+2.$
Up the path now, node $\ell-1$ has $b-\ell+2$ leaves of frequency 1 and a neighbor above it in the
path also of frequency 1. Since there is a tie we arbitrarily
visit one leaf once more
before proceeding upwards in the path thus leaving node $\ell-1$ with frequency
$b-\ell+6$. We repeat this pattern alternating between odd and even
thus obtaining the expression:
\[
\mbox{freq}(i)=\left\{\begin{tabular}{ll}
  $(b-i+1)(b-i+3)+2$  &\mbox{ if $i$ is odd }\\
  $b-i+6$  &\mbox{ if $i$ is even }
\end{tabular}
\right.
\]
Here we have the following crucial property:
\begin{claim}
The parent of a node in the path has higher frequency than its child in the path, i.e.
$\mbox{freq}(i)>\mbox{freq}(i+2)$ for $i\leq\ell-2$
\end{claim}
\begin{proof}
Substituting in the expression above we have,
\begin{eqnarray*}
\mbox{freq}(i)-\mbox{freq}(i+2)
                                              &=&(b-i+1)\,4\,>\,0 \quad\quad\quad\quad \mbox{ if $i$ is odd }\\
\mbox{freq}(i)-\mbox{freq}(i+2)&=&(b-i+6)-(b-i+4) \,=\,2 \,>\, 0 \quad\mbox{ if $i$ is even }
\end{eqnarray*}
and the property holds as claimed.
\qed
\end{proof}
Observe that when node $i$ in the path is reached, node $i-1$ has larger frequency than
node $i+1$ hence the robot visits its leaf children $\mbox{freq}(i+1)$ many times, and ends up with a
frequency of the number of children times the frequency of node $i+1$ plus two.
%
%
\[
\mbox{freq}(i)=\left\{\begin{tabular}{ll}
 $b\,(b+2)+2$  &\mbox{ if $i=1$}\\
  $(b-i+1)(b-i+3)+3$  &\mbox{ if $1<i\leq\ell$ is odd }\\
  $(b-i+1)[(b-i)(b-i+4)+2]+3$  &\mbox{ if $1<i\leq \ell$ is even }\\
  $(b+1)[(b-\ell+1)(b+2)+3]+2$  &\mbox{ if $i=\ell+1$}
\end{tabular}
\right.
\]
Observe that the property that $\mbox{freq}(i)>\mbox{freq}(i+2)$ holds just as before for $i\leq\ell-3$.
Let us establish this as an invariant. First consider  $\mbox{freq}(i)$ as a polynomial in $b$ and we have
that the degrees form the sequence $1,2,3,2,3,\ldots,2,3$ as shown in column 3 of Table 1, and using
the claim above we get:
\begin{invariant}\label{inv}
If a node in the path has parent and child in the path with frequency of degree 3, then
the frequency of the parent is higher than the frequency of the child. I.e. if
$\mbox{deg$($freq$(i-1))$}= \mbox{deg$($freq$(i+1))$}$  then $\mbox{freq}(i-1)>\mbox{freq}(i+1)$ for $i\leq \ell-2$.
\end{invariant}
Using this invariant and Table 1 we can study the change in degrees as we proceed up the path, starting from node $\ell+1$ in column 4.
This node has frequency of degree 3 and its upper neighbor has frequency of
degree 2 so its new frequency is $\Theta(b^3)$ and its
frequency degree remains unchanged. Then we move to node $\ell$ whose
two neighbors have frequency degree 3 hence its new frequency
is $\Theta(b^4)$ as shown in column 4 of Table 1. At this point the degree of the child of the current node is higher than the degree
of the parent in the path (recall that Invariant \ref{inv} holds only for $i\leq\ell-2$) and hence we continue upward until we reach node $\ell-2$ whose
child and parent have degree 3. This node then increases its frequency to degree 4 and moves
down as per Invariant \ref{inv}.

\noindent Note that the last node visited in the upward path is given by the invariant. Specifically the degree of
the frequency of node $i$ at pass $j$ for $i\geq 1$ and $j\geq 3$ is
{\small
\[
\mbox{Deg}_{i,j}=\left\{\begin{tabular}{ll}
  $2$  &\mbox{if $i$ is even and $j<\ell-i$  }\\
  $3$  &\mbox{if $i$ is odd and $j<\ell-i$  }\\
  $j-\ell-i$ &\mbox{if $i$ is even, $j$ is odd and $j\geq \ell-i$} \\
  $j-\ell-i+1$ &\mbox{if $i$ is even, $j$ is even and $j\geq \ell-i$} \\
  $j-\ell-i+2$ &\mbox{if $i$ is odd, $j$ is odd and $j\geq \ell-i$} \\
  $j-\ell-i+1$ &\mbox{if $i$ is odd, $j$ is even and $j\geq \ell-i$} \\
  $j$ &\mbox{if $i=\ell+1$, $j$ is odd} \\
  $j-1$ &\mbox{if $i=\ell+1$, $j$ is even}
 \end{tabular}
\right.
\]
}

 \begin{table}
\begin{center}
\resizebox{0.53\linewidth}{!}{%
\begin{tabular}{lllllllllllll}
 	$\downarrow$& $\uparrow$& $\downarrow$& $\uparrow$&$\downarrow$& $\uparrow$&$\downarrow$& $\uparrow$&$\downarrow$& $\uparrow$&$\downarrow$& $\uparrow$&$\downarrow$\\
 $\quad$\vspace{-1mm}\\
	{\bf 1} $\,\,\,\,$&	1$\quad$& 1$\quad$&1$\quad$&1$\quad$&1$\quad$&1$\quad$&1$\quad$&1$\quad$&1$\quad$&1$\quad$&1$\quad$&1$\quad$\\
	\bf 0&	{\bf 2}&	2&	2&	2&		2&		2&		2&		2&		2&		2&		2&		2\\
	\bf 1&	\bf 1&	\bf 3&	3&		3&		3&		3&		3&		3&		3&		3&		3&		3\\
	\bf 0&	\bf 2&	\bf 2&	2&		2&		2&		2&		2&		2&		2&		2&		\bf 4&	\bf 4\\
      	\bf 1&	\bf 1&	\bf 3&	3&		3&		3&		3&		3&		3&		3&		3&		\bf 3&	\bf 5\\
  	\bf 0&    	\bf 2&   	\bf 2&	2&		2&		2&		2&		2&		2&		\bf 4&	\bf 4&	\bf 4&	\bf 4\\
	\bf 1&	\bf 1&	\bf 3&	3&		3&		3&		3&		3&		3&		\bf 3&	\bf 5&	\bf 5&	.\\
	\bf 0&	\bf 2&	\bf 2&	2&		2&		2&		2&		\bf 4&	\bf 4&	\bf 4&	\bf 6&	\bf 6&	.\\
	\bf 1&	\bf 1&	\bf 3&	3&		3&		3&		3&		\bf 3&	\bf 5&	\bf 5&	\bf 7&	\bf 7&	.\\
	\bf 0&	\bf 2&	\bf 2&	2&		2&		\bf 4&	\bf 4&	\bf 4&	\bf 6&	\bf 6&	\bf 8&	\bf 8\\
	\bf 1&	\bf 1&	\bf 3&	3&		3&		\bf 3&	\bf 5&	\bf 5&	\bf 7&	\bf 7&	\bf 9&	\bf 9\\
	\bf 0&	\bf 2&	\bf 2&	{\bf 4}&	\bf 4&	\bf 4&	\bf 6&	\bf 6&	\bf 8&	\bf 8&	\bf 10&	\bf 10\\
	\bf 1&	\bf 1&	\bf 3&	\bf 3&	\bf 5&	\bf 5&	\bf 7&	\bf 7&	\bf 9&	\bf 9&	\bf 11&	\bf 11\\
	\bf 0&	\bf 2&	\bf 2&	\bf 4&	\bf 4&	\bf 6&	\bf 6&	\bf 8&	\bf 8&	\bf 10&	\bf 10&	\bf 12\\
	\bf 1&	\bf 1&	\bf 3&	\bf 3&	\bf 5&	\bf 5&	\bf 7&	\bf 7&	\bf 9&	\bf 9&	\bf 11&	\bf 11
\end{tabular}}
\end{center}
\caption{Evolution of the degrees of the frequencies. Each column represents the
frequency degrees during a traversal of the path. The arrow indicates the direction of traversal.
In a downward trajectory the last node in the path is always reached, whereas in an upward trajectory
we indicate with bold the extent of the nodes visited in each pass. Rows are numbered $0,1,\ldots,\ell+1$ while columns are numbered $1,2,\ldots,\ell-2$.}
\vspace{-3mm}
\end{table}

The relation above has as exit condition when $j=\ell-2$. At this point the upward path reaches the top node and finally visits the leaf of the root
left unvisited in the very first step (see Fig.~\ref{fig:caterpillar}). The largest degree is attained in the last node of the path which has frequency of degree $j$.
Hence if we set $b=\sqrt{n}$ we obtain a polynomial of degree $\left(\sqrt{n}\right)^{\ell-2}=\Theta((n^{1/2})^{\sqrt{n}-2})=\Theta(n^{\sqrt{n}/2})$
thus proving Theorem~\ref{co:ub.LFV-v}.

\section{Conclusions}
In this paper we give (1) an exponential lower bound for the
worst case for LRV-v of triangulations (2) a quadratic lower bound for
the worst case for LRV-e of triangulations (3) an exact bound on the maximum degree
difference between two neighboring nodes in LFV-v (4) a quadratic lower bound for LFV-v
in graphs of degree 3 and, most importantly, (5) a superpolynomial lower bound for the
worst-case of LFV-v.

We conjecture that for graphs of maximum degree 3, the performance of LFV-v is quadratic
and its average coverage time is linear.

\bibliographystyle{plain}
\bibliography{lit,bibliography,mclurkin-bibliography}

\end{document}